\documentclass{llncs}

\usepackage{amssymb}
\usepackage{lineno}
\usepackage{graphicx}

\newcommand{\quer}{Q}  
\newcommand{\quert}{Q'}
\newcommand{\quertt}{Q''} 
\renewcommand{\Re}{\mathbb{R}}

\title{Coloring Planar Homothets and Three-Dimensional Hypergraphs}

\author{Jean Cardinal \and Matias Korman}
\institute{
Universit\'e Libre de Bruxelles (ULB)\\
Brussels, Belgium\\
{\tt\{jcardin,mkormanc\}@ulb.ac.be}
}

\begin{document}
\maketitle
\linenumbers

\begin{abstract}
We prove that every finite set of homothetic copies of a given compact and convex body in the plane can be colored with four colors so that any point covered by at least two copies is covered by two copies with distinct colors. This generalizes a previous result from Smorodinsky (SIAM J. Disc. Math. 2007). 
Then we show that for any $k\geq 2$, every three-dimensional hypergraph can be colored with $6(k-1)$ colors so that every hyperedge $e$ contains $\min\{ |e|,k \}$ vertices with mutually distinct colors. This refines a previous result from Aloupis {\em et al.} (Disc. \& Comp. Geom. 2009).  As corollaries, we improve on previous results for conflict-free coloring, $k$-strong conflict-free coloring, and choosability. Proofs of the upper bounds are constructive and yield simple, polynomial-time algorithms.
\end{abstract}

\section{Introduction}

The well-known graph coloring problem has several natural generalizations to hypergraphs. A rich literature exists on these topics; in particular, the two-colorability of hypergraphs (also known as property B), has been studied since the sixties. In this paper, we concentrate on coloring geometric hypergraphs, defined by simple objects in the plane. Those hypergraphs serve as models for wireless sensor networks, and associated coloring problems have been investigated recently. This includes conflict-free colorings
~\cite{shakharcf,HPS05}, and covering decomposition problems~\cite{pachtoth,pachindecomp,GV09}.

Smorodinsky~\cite{Smo07} investigated the chromatic number of geometric hypergraphs, defined as the minimum number of colors required to make every hyperedge non-monochromatic. He considered hypergraphs induced by a collection $S$ of regions in the plane, whose vertex set is $S$, and the hyperedges are all subsets $S'\subseteq S$ for which there exists a point $p$ such that $S'= \{ R \in S: p\in R \}$. He proved the following result.
\begin{theorem}\label{thm_4col}
\begin{itemize}
\item Any hypergraph induced by a family of n simple Jordan regions in the plane such that the union complexity of any $m$ of them is given by $u(m)$ and $u(m)/m$ is non-decreasing is $O(u(n)/n)$-colorable so that no hyperedge is monochromatic. In particular, any finite family of pseudodisks can be colored with $O(1)$  colors.
\item Any hypergraph induced by a finite family of disks is 4-colorable
\end{itemize}
\end{theorem} 
Later, Aloupis, {\em et al.}~\cite{ACCLS09} considered the quantity $c(k)$, defined as the minimum number of colors required to color a given hypergraph, such that every hyperedge of size $r$ has at least $\min \{ r, k\}$ vertices with distinct colors. For hypergraphs induced by a collection of regions in the plane, such that no point is covered more than $k$ times (a $k$-fold packing), this number corresponds to the minimum number of (1-fold) packings into which we can decompose this collection. It generalizes the usual chromatic number, equal to $c(2)$. They proved the following.
\begin{theorem}
\label{thm_ck}
Any finite family of pseudodisks in the plane can be colored with $24k+1$ colors in a way that any point covered by $r$ pseudodisks is covered by $\min \{r, k\}$ pseudodisks with distinct colors.
\end{theorem}

\paragraph{Our results.}
We show in Section~\ref{sec_dual} that the second statement of Theorem~\ref{thm_4col} actually holds for homothets of any compact and convex body in the plane. The proof uses a lifting transformation that allows us to identify a planar graph, such that every hyperedge of the initial hypergraph contains an edge of the graph. The result then follows from the Four Color Theorem. 

We actually give two definitions of this graph: one is based on a weighted Voronoi diagram construction, while the other relates to Schnyder's characterization of planar graphs. Schnyder showed that a graph is planar if and only if its vertex-edge incidence poset has dimension at most $3$~\cite{schnyder}. In Section~\ref{sec_3D}, we show that the chromatic number $c(k)$ of three-dimensional hypergraphs is at most $6(k-1)$. This improves the constant of Theorem~\ref{thm_ck} for this special case, which includes hypergraphs induced by homothets of a triangle. 

In Section~\ref{sec_lb}, we give a lower bound for all the above problems. 
 Finally, in Section~\ref{sec_appl}, we give some corollaries of these results involving other types of colorings, namely conflict-free and $k$-strong conflict-free colorings, and choosability. 

\paragraph{Definitions.}
We consider hypergraphs defined by {\em ranges}, which are compact and convex bodies $\quer \subset \Re^d$ containing the origin. 
The {\em scaling} of $\quer$ by a factor $ \lambda\in\Re$ is the set $\{\lambda x : x\in \quer \}$. Note that, the scaling of $\quer$ with $\lambda=-1$ is the reflection of $\quer$ around the origin. 
The {\em translate} of $\quer$ by a vector $t\in\Re^d$ is the set $\{x+t : x\in \quer \}$.
The {\em homothet} of $\quer$ of {\em center} $t$ and {\em scaling} $\lambda$ is the set $\{\lambda x + t : x\in \quer \}$ and is denoted by $\quer(t,\lambda)$. Given a finite collection $S$ of points in $\Re^d$, the {\em primal hypergraph} defined by these points and a range $\quer$ has $S$ as vertex set, and $\{ S\cap \quer' : \quer'\mathrm{\ homothet\ of\ }\quer \}$ as hyperedge set. Similarly, the {\em dual hypergraph} defined by a finite set $S$ of homothets of $\quer$ has $S$ as vertex set, and the hyperedges are all subsets $S'\subseteq S$ for which there exists a point $p\in\Re^d$ such that $S'= \{ R \in S: p\in R \}$ (i.e., the set of ranges of $S$ that contain $p$). While we give these definitions for an arbitrary dimension $d$, we will be mostly concerned by the case $d=2$. 

For a given range $\quer$, the chromatic number $c_{\quer} (k)$ is the minimum number $c$ such that every primal hypergraph (induced by a set of points) can be colored with $c$ colors, so that every hyperedge of size $r$ contains $\min \{r, k\}$ vertices with mutually distinct colors. Similarly, the chromatic number $\bar{c}_{\quer} (k)$ is the smallest number $c$ such that every dual hypergraph (induced by a set of homothets of $\quer$) can be $c$-colored so that every hyperedge of size $r$ contains $\min \{r, k\}$ vertices with mutually distinct colors. In what follows, we refer to these two coloring problems as {\em primal} and {\em dual}, respectively. Such colorings are called {\em polychromatic}\footnote{The term {\em $k$-colorful} is also used in the literature~\cite{Smo07}.}.

\section{Coloring Primal Hypergraphs}\label{sec_primal}
As a warm-up, we consider the primal version of the problem for $k=2$. 
Given a set of points $S\subset\mathbb{R}^2$ and a two-dimensional range $\quer$, the {\em Delaunay graph} of $S$ induced by $\quer$ is the graph $G_Q(S)=(S,E)$ with  $S$ as vertex set~\cite{fortune}. For any two points $p,q\in S$, $pq\in E$ if and only if there exists a homothet $\quert$ of $\quer$ such that $\quert \cap S =\{p,q\}$. Note that, the Delaunay graph induced by disks in the plane corresponds to the ordinary Delaunay triangulation, which is planar. In fact, planarity holds for many ranges.
\begin{lemma}~\cite{BCCS08c,sarioz}\label{lem_planar}
For any convex range $\quer\subseteq \Re^2$ and set of points $S$, $G_Q(S)$ is planar.
\end{lemma}
Previously published versions of this result required that the points of $S$ are in general position (that is, no four points of $S$ are on the boundary of a range). The generalization to any point set was done by Bose {\em et al.}~\cite{BCCS08c}. Whenever a homothet $\quert$ contains more than $3$ points on its boundary, the edges $uv,uw$ and $vw$ are added to $G_Q(S)$, where $u,v$, and $w$ are the three lexicographically smallest points of $S\cap \quert$. With this definition, they showed that planarity holds for any compact and convex range. The compactness requirement was afterwards removed by Sarioz~\cite{sarioz}.

\begin{figure}[tb]
\center
\includegraphics[width=0.4\textwidth]{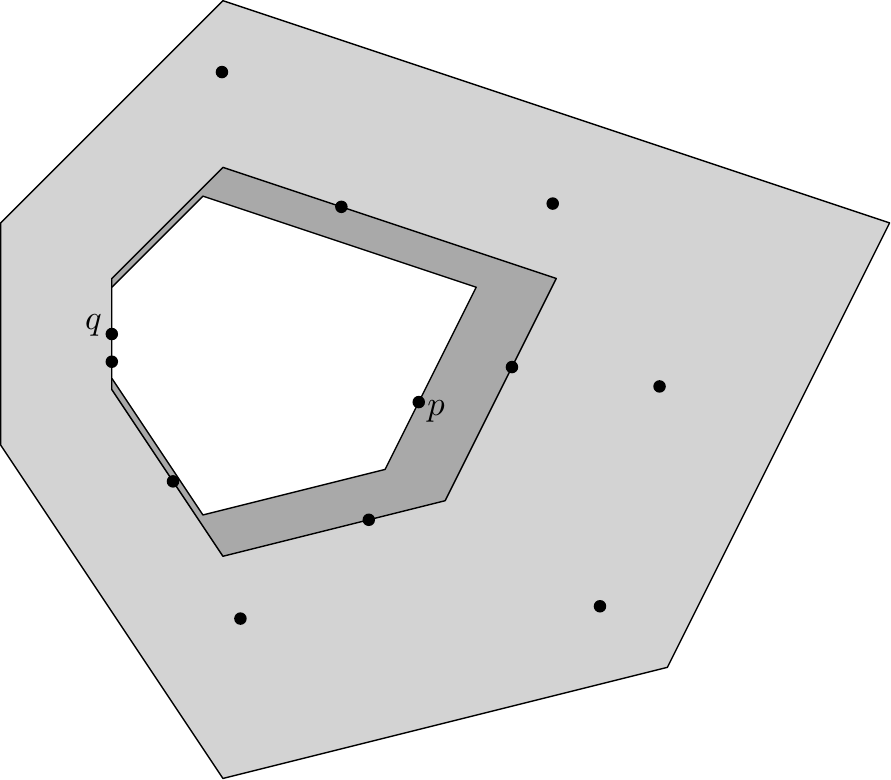}
\caption{Proof of Lemma~\ref{lem_edge}: given a homothet $\quert$ (light grey), we shrink it until further shrinking will have less than two points (dark grey). Afterwards we keep shrinking while remaining tangent to a point $q$ on the boundary until the point in the interior of the range (if any) reaches the boundary. The resulting range $\quertt$ is depicted in white. }
\label{fig_shrink}
\end{figure}

\begin{lemma}
\label{lem_edge}
For any homothet $\quert$ containing two or more points of $S$, there exist $p,q\in S \cap \quert$ such that $pq\in E$.
\end{lemma}
\begin{proof}
Let $\quert$ be a homothet of center $c_0$ and scaling $\lambda_0$ that contains two or more points of $S$. We shrink it continuously keeping the same center; let $\lambda_{\min}$ be the smallest scaling such that $\quer(c_0,\lambda_{\min})$ has two (or more) points of $S$. If $\quer(c_0, \lambda_{\min})$ contains exactly  two points $p,q\in S$, we have $pq\in E$ by definition of $G_Q(S)$. However, we might have some kind of degeneracy in which $\quer(c_0,\lambda_{\min})$ contains three (or more) points of $S$. Observe that this can only happen if there are two or more points on the boundary and possibly an interior point. 

First consider the case in which there exists a point $p\in S$ in the interior of $\lambda_{\min}$. Pick any point $q\in S$ on the boundary of $\quer(c_0,\lambda_{\min})$ and shrink the homothet 
 remaining tangent to $q$ until $p$ reaches the boundary  (see Figure \ref{fig_shrink}). After this shrinking process, both $p$ and $q$ will be on the boundary of the new homothet. Moreover, any other point that was previously in $\quer(c_0,\lambda_{\min})$ either remains on the boundary or is not in the range anymore. 
 
That is, we can always shrink a range $\quert$ to another range $\quertt \subseteq \quert$ that contains two or more points on its boundary and none in the interior. Hence, by the result of \cite{BCCS08c} we know that there will be an edge between the two lexicographically smallest points of $S \cap \quertt$. 
\qed\end{proof}

\begin{theorem}\label{theo_primal}
For any two-dimensional range $\quer$, we have $c_\quer(2)\leq4$.
\end{theorem}
\begin{proof}
By Lemma \ref{lem_planar}, $G_Q(S)$ is planar, hence 4-colorable. By Lemma \ref{lem_edge}, any homothet $\quert$ containing two or more points of $S$ must contain $p,q\in S\cap \quert$ such that $pq\in E$. In particular, these points cannot have the same color assigned, hence $\quert$ cannot be monochromatic.
\qed\end{proof}
The proof yields an $O(n^2)$-time algorithm. The bound of Theorem~\ref{theo_primal} is tight for a wide class of ranges (see Section \ref{sec_lb}). 

\section{Coloring Dual Hypergraphs}\label{sec_dual}
 
In this section we describe a similar approach for the dual variant of the problem in the plane. Recall that in the dual problem, we are given a set $S$ of compact and convex homothets of $\quer$, and we are interested in coloring the elements of $S$ so that any point of the plane covered at least twice is covered by two homothets of different colors. For simplicity, we first suppose that no range of $S$ is contained in another; we will show how to remove this assumption afterwards. 

In order to solve this problem, we lift the two-dimensional ranges to a three dimensional space. We map the homothet of $\quert$ of center $(x,y)$ and scaling $\lambda$ to the three dimensional point  $\rho (\quert )=(x, y, \lambda)\in \Re^3$. Given a set $S$ of homothets of $\quer$, we define $\rho (S) = \{\rho (\quert) : \quert \in S\}$ as the set containing the images of the ranges in $S$. For any point $p=(x,y,d)$, we associate the three dimensional range $\pi(p)$ as the cone with apex at $(x,y,d)$ such that the intersection with the horizontal plane of height $z$ is $Q(0, z-d)$ (if $z\geq d$) or empty (if $z<d$), where $\quer^*=\quer(0,-1)$ is the reflection of $\quer$ about its center.
 Note that the cone $\pi (p)$ so defined is convex (see Figure \ref{fig_trans}). We define the {\em downward cone} $\pi^* (p)$ as the centrally symmetric image of $\pi(p)$ through $p$. By symmetry, we observe the following: 

\begin{figure*}[t]
  \centering
	\begin{tabular}{cc}

		\begin{minipage}{0.5\hsize}
				\includegraphics[width=\textwidth]{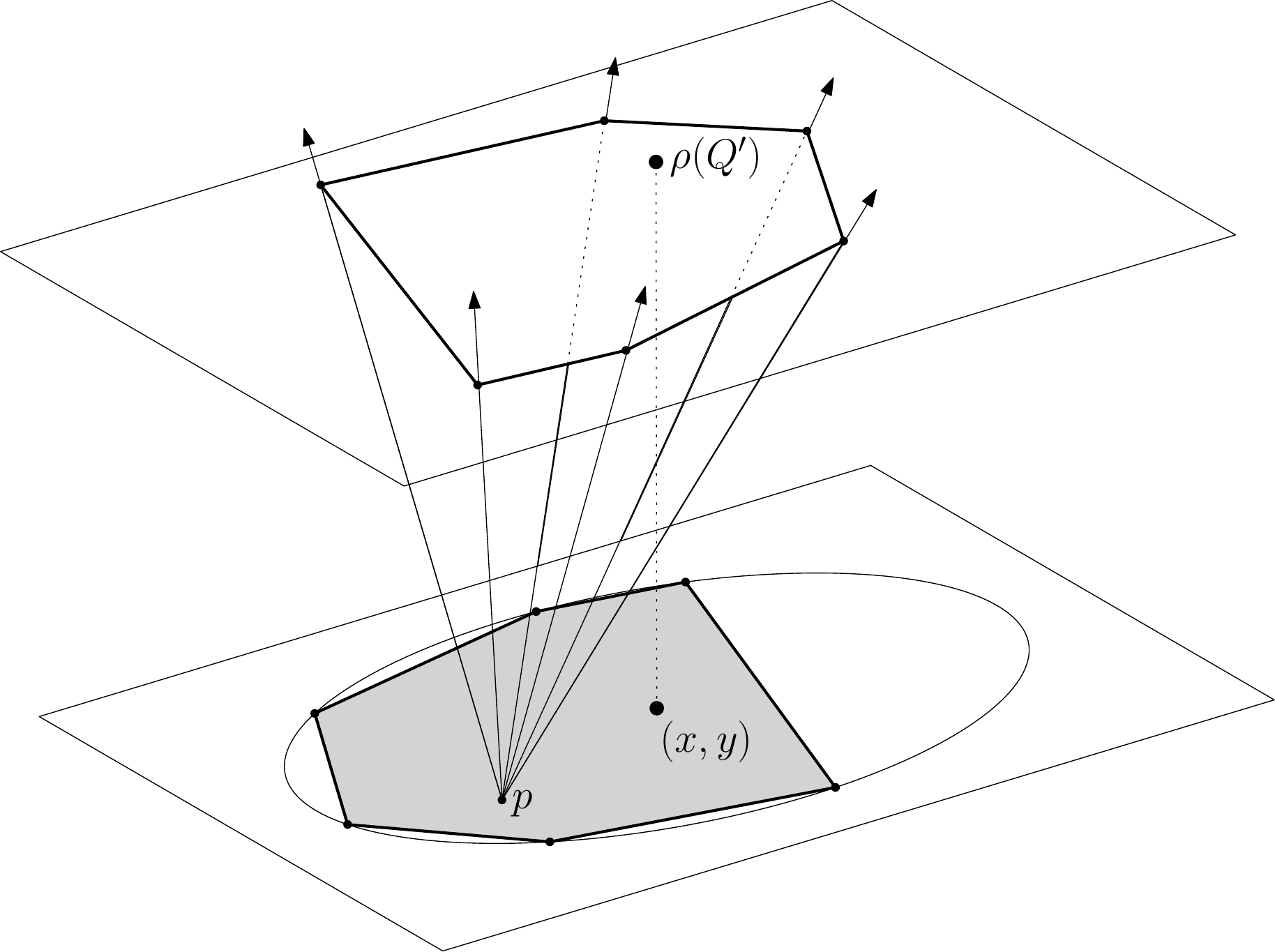}
				\caption{Mapping a homothet $\quert$ (in grey) to a point in $\rho (\quert ) \in\Re^3$, and a point $p$ to a cone $\pi (p)$.}
				\label{fig_trans} 
		\end{minipage}
		
		\begin{minipage}{0.5\hsize}
			\begin{center}
				\includegraphics[width=\textwidth]{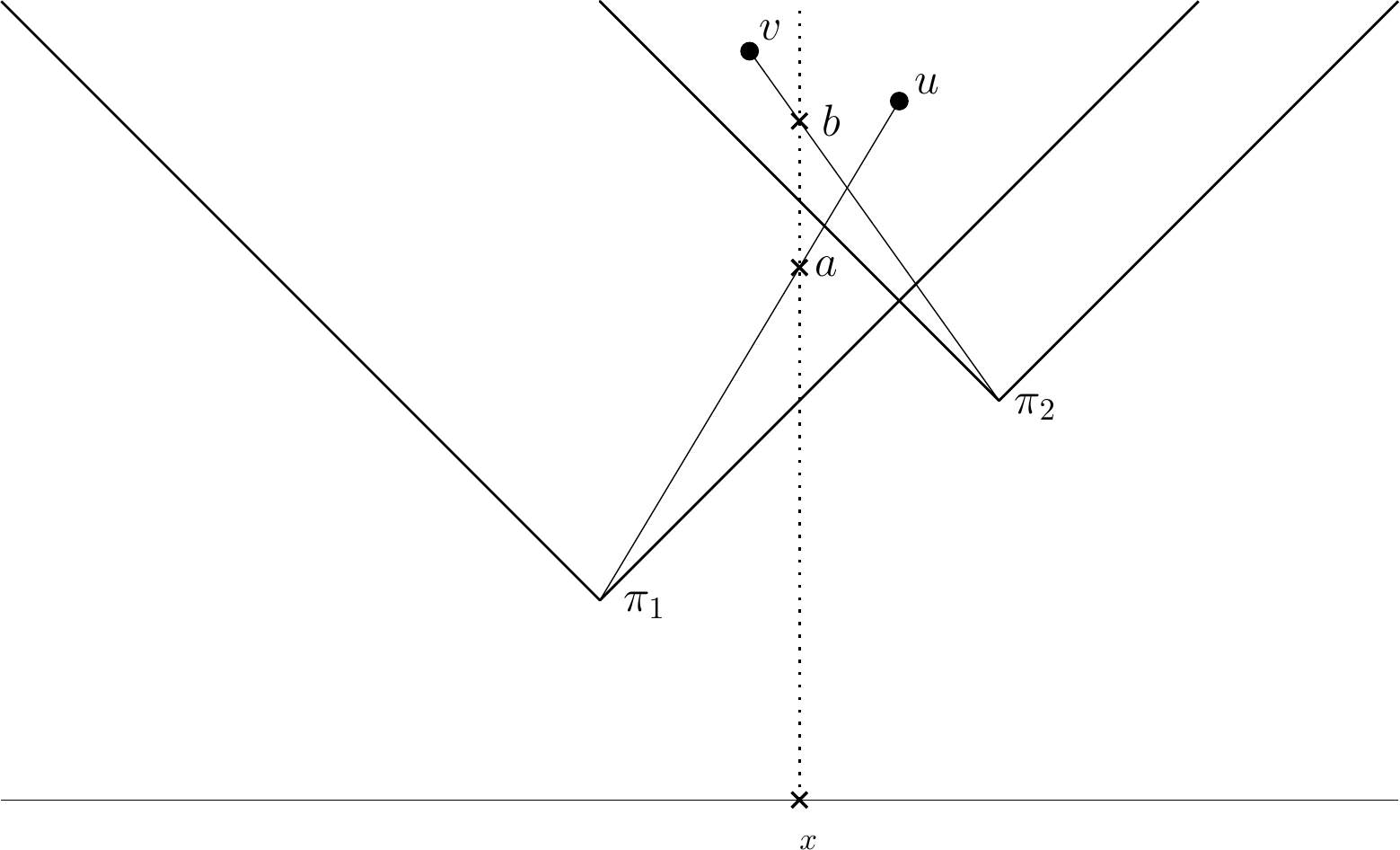}
				\caption{A crossing in the embedding implies that a cone contains three points.}
				\label{fig_planar}
			\end{center}
		\end{minipage}		

	\end{tabular}
\end{figure*}
\begin{lemma}\label{lem_dual}
For any $p,q\in\Re^3$, we have $q\in\pi (p)\Leftrightarrow p\in\pi^*(q)$. Moreover, for any point $(x, y)\in \Re^2$ and range $\quert$, $(x, y)\in \quert \Leftrightarrow \rho(\quert ) \in \pi((x, y, 0))$.
\end{lemma}

It follows that any coloring of $\rho(S)$ with respect to the conic ranges $\pi$ is a valid coloring of $S$. Let $G_\pi(\rho(S))$ be the Delaunay graph in $\Re^3$ with cones as ranges. That is, the vertex set of $G_\pi(\rho(S))$ is  $S$ and two ranges $\quert ,\quertt$ of $S$ are adjacent if and only if there exists a point $p\in\Re^3$ such that $\pi (p) \cap \rho (S) = \{ \rho (\quert ), \rho (\quertt )\}$. We claim that $G_\pi(\rho(S))$ satisfies properties similar to those of Lemmas~\ref{lem_planar} and \ref{lem_edge}. In order to prove so, we first introduce some inclusion properties.

\begin{lemma}\label{lem_inclus}
For any $p\in\Re^3$, $q\in \pi (p)$ and $m$ on the line segment $pq$, we have $\pi (q) \subseteq \pi (p)$
and $q \in \pi (m)$.
\end{lemma}
\begin{proof}
Observe that the projections of the cones $\pi (p)$ and $\pi (q)$ on any vertical plane (i.e., any plane of equation $ax+by+c=0$) are two-dimensional cones; that is, the set of points above two halflines with a common origin. Moreover, the slope of the halflines only depends on $a,b$ and $\quer$. Let $z_q$ be the $z$-coordinate of $q$, and consider the intersections of the cones $\pi (p)$ and $\pi (q)$ with a horizontal plane $\Pi$ of $z$-coordinate $t \geq z_q$. We get two homothets of $\quer^*$, say $\quer^*_p$ and $\quer^*_q$. We have to show that $\quer^*_q \subseteq \quer^*_p$ for any $t$.

Suppose otherwise. There exists a vertical plane $\Pi'$ for which the projection of $\quer^*_q$ on $\Pi'$ is not included in the projection of $\quer^*_p$. To see this, we can find a common tangent to $\quer^*_p$ and $\quer^*_q$ in $\Pi$, slightly rotate it so that it is tangent to $\quer^*_q$ only, then pick a plane that is orthogonal to that line. But the projections of $\pi (p)$ and $\pi (q)$ on $\Pi'$ are two cones with parallel bounding halflines, thus the projection of the apex of $\pi (q)$ cannot be in that of $\pi (p)$, a contradiction.

The proof of the second claim is similar. We know that $q\in \pi (p)$, hence from Lemma~\ref{lem_dual}, $p\in\pi^* (q)$. Now from the convexity of $\pi^* (q)$, we have that $m\in \pi^*(q)$.  Using again Lemma~\ref{lem_dual}, we obtain $q\in\pi (m)$.
\qed\end{proof}

\begin{lemma}\label{lem_dualplanar}
The graph $G_\pi(\rho(S))$ is planar. 
\end{lemma}
\begin{proof}
By definition of $E(S)$, we know that for every edge $\quert\quertt\in E$ there exists $p\in\Re^3$ such that $\pi (p) \cap \rho (S) = \{ \rho (\quert ), \rho (\quertt )\}$. We draw the edge $\quert\quertt$ as the projection (on the horizontal plane $z=0$) of the two line segments connecting respectively $\rho (\quert )$ and $\rho (\quertt )$ to $p$.

First note that crossings involving two edges with a common endpoint can be eliminated by rerouting the two polygonal lines at their intersection point. Thus it suffices to show that this embedding has no crossing involving vertex-disjoint edges. Consider two edges $uu'$ and $vv'$, and their corresponding witness cones $\pi_1\ni u, u'$ and $\pi_2\ni v, v'$. By definition of witness, each cone must contain exactly  two points. In particular, we have $u\not\in\pi_2$ and $v\not\in\pi_1$.

Suppose that the projections of the segments connecting $u$ with the apex of $\pi_1$ and $v$ with the apex of $\pi_2$ cross at a  point $x$ (other than the endpoints). Consider the vertical line $\ell$ that passes through $x$: by construction, this line must intersect with both segments at points $a$ and $b$, respectively. Without loss of generality we assume that $a$ has lower $z$ coordinate than $b$ (see Figure~\ref{fig_planar}). 

From the convexity of $\pi_1$, we have $a\in \pi_1$. From Lemma~\ref{lem_inclus}, we have $v\in \pi (b)$, $b\in \pi (a)$, and $\pi (a)\subseteq \pi (b)$. In particular, we have $v\in \pi (b) \subseteq \pi (a) \subseteq\pi_1$, which contradicts $v\not\in \pi_1$.
\qed\end{proof}

\subsection{Alternative construction via weighted Voronoi diagrams}\label{sec:wVD}

We introduce an alternative definition of $G_\pi(\rho(S))$ so as to prove its planarity. For any point $p$, we define its distance to $q$ as $d(p,q)=\min\{\lambda \geq 0 | q\in \quer(p,\lambda)\}$. That is, the smallest possible scaling so that a range of center $p$ contains $q$. This distance is called the {\em convex distance function} from $p$ to $q$ (with respect to $\quer$). 

Given a set $S=\{\quer_1, \ldots, \quer_n\}$ of homothets of $\quer$, we construct an additively weighted Voronoi diagram $V_\quer(S)$ with respect to the convex distance function~\cite{fortune}. Let $c_i$ and $\lambda_i$ be the center and scaling of $\quer_i$. Then $V_\quer(S)$ has $\{c_1,\ldots c_n\}$ as the set of sites, and each site $c_i$ is given the weight $-\lambda_i$. The additively weighted Voronoi diagram for this set of sites has a cell for each site $c_i$, defined as the locus of the points $p$ of the plane whose weighted distance $d(p,c_i)-\lambda_i$ to $c_i$ is the smallest among all sites. The dual graph for this Voronoi diagram has an edge between any two sites whose cells share a boundary. 

In the following we show that the dual of $V_\quer(S)$ is $G_\pi(\rho(S))$. Let $p = (x,y)\in\Re^2$ be any point covered by one or more ranges of $S$. We denote by $(p,z)$ the point of $\Re^3$ of coordinates $(x,y,z)$, for any $z\in\Re$. From Lemma~\ref{lem_dual}, the points of $\rho (S)$ contained in $\pi(p,0)$ are exactly the ranges of $S$ that contain $p$. We translate the cone $\pi (p,0)$ vertically upward; in this lifting process, the points of $\rho(S)\cap \pi(p,0)$ will leave the cone. For any  homothet $\quert\in S$ such that $p\in \quert$, we define $z_{\quert}(p)$ as the height in which $\rho(\quert)$ is on the boundary of the cone $\pi(p,z_{\quert}(p))$. 

\begin{figure*}[t]
  \centering
	\begin{tabular}{cc}

		\begin{minipage}{0.5\hsize}
				\includegraphics[width=\textwidth]{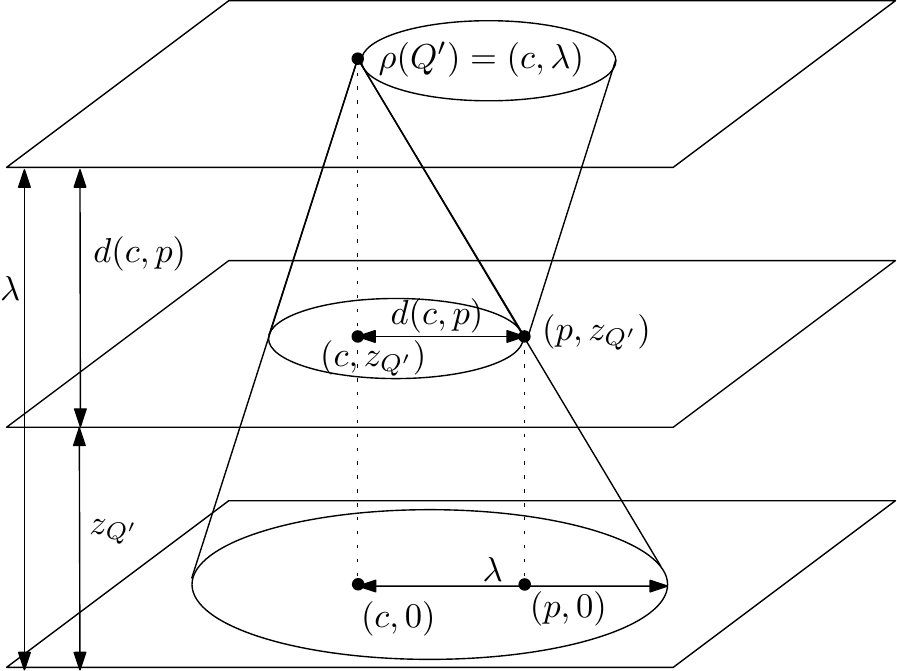}
				\caption{Proof of Lemma \ref{lem_height}. 
				}
				\label{fig_vorono} 
		\end{minipage}
		
		\begin{minipage}{0.5\hsize}
			\begin{center}
				\includegraphics[width=\textwidth]{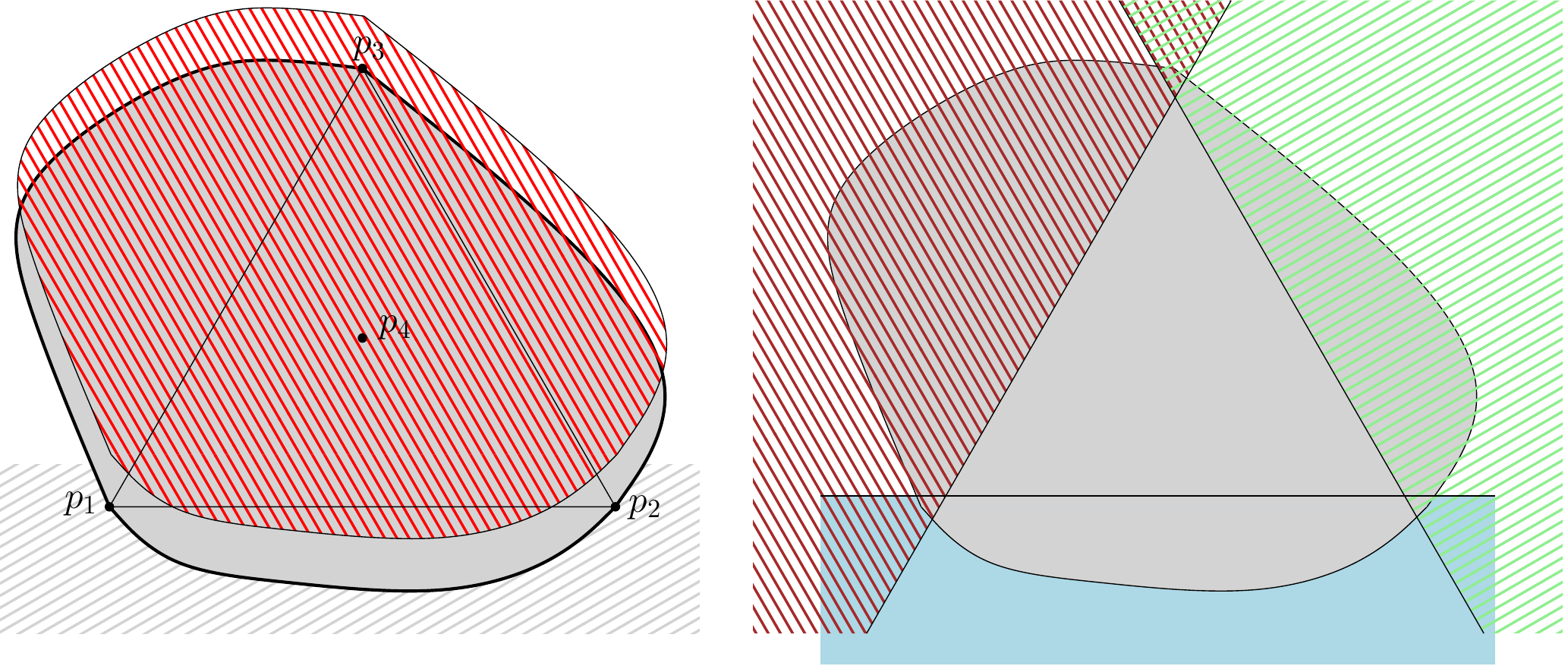}
				\caption{(left): Any two points of the set can be covered by either translating or scaling the range $\quer$, hence $c_\quer(2)\geq 4$. (right): Analogous construction for the dual problem.}
				\label{fig_ranges}
			\end{center}
		\end{minipage}		

	\end{tabular}
\end{figure*}

\begin{lemma}\label{lem_height}
For any homothet $\quert$ of center $c$ and scaling $\lambda$ and $p\in\quert$, $z_{\quert}(p)=\lambda - d(c,p)$.
\end{lemma}
\begin{proof}
Consider the cone $\pi(p,z_{\quert}(p))$ and the halfplane $z=z_{\quert}(p)$.  By Lemma \ref{lem_dual}, the fact that $\rho(\quert)$ is on the boundary of $\pi(p,z_{\quert}(p))$ is equivalent to $(p,z_{\quert}(p))$ being on the boundary of the downwards cone $\pi^*(\rho(\quert))=(c,\lambda)$. Observe that the distance between points $(c,z_{\quert}(p))$ and $(p,z_{\quert}(p))$ is exactly $d(c,p)$. Consider the plane that passes through the points $(c,0)$, $(p,0)$, and $(c,\lambda)$: we have two similar triangles whose bases have lengths $\lambda$ and $d(c,p)$, respectively. Since the height of the large triangle is $\lambda$, we conclude that the height of the smaller one must be $d(c,p)$ (see Figure \ref{fig_vorono}). That is, the difference in the $z$ coordinates between the points $(c,z_{\quert}(p))$ and $(c, \lambda)$ is $d(c,p)$. Since, by definition of $\rho$, the difference in $z$ coordinates between $(c,0)$ and $\rho(\quert)$ is exactly $\lambda$, we obtain the equality $\lambda=z_{\quert}(p)+d(c,p)$ and the Lemma follows.
\qed\end{proof}

This shows the duality between the weighted Voronoi diagram and the graph $G_\pi(\rho(S))$: let $p\in \Re^2$ be any point in the plane covered by at least one range of $\quert$. Consider the cone $\pi(p,0)$ and lift it continuously upward. The last point of $\rho(S)\cap \pi(p,0)$ to leave the cone will be one with highest $z_{\quert}(p)$. By Lemma \ref{lem_height}, it will be the homothet $\quert$ of center $c'$ and scaling $\lambda'$ that has the {\em smallest} $d(c',p)-\lambda'$. Observe that this is exactly how we defined the weights of the sites, hence $\quert$ being the last range in the cone is equivalent to $c'$ being the {\em closest} site of $p$ in $V_\quert(S)$. This can be interpreted as shrinking simultaneously all ranges until $p$ is only covered by its closest homothet $\quert$. This shrinking process is simulated in our construction through the $z$ coordinate. 

\begin{lemma}\label{lem_dualVorDel}
The dual graph of $V_\quert(S)$ is $G_\pi(\rho(S))$.
\end{lemma}
\begin{proof}
Let $p=(x,y)\in \Re^2$ be a point on a bisector of $V_\quer(S)$ between sites $c_1$ and $c_2$ (corresponding to ranges $\quer_1$ and $\quer_2$ of scaling $\lambda_1$ and $\lambda_2$, respectively). By definition, we have that $d(c_1,p)-\lambda_1=d(c_2,p)-\lambda_2$ and $d(c',p)-\lambda' > d(c_1,p)-\lambda_1$ for all other homothets $\quert \in S$ of center $c'$ and scaling $\lambda'$. 

Let $w_{\min}=d(c_1,p)-\lambda_1$ and consider the cone $\pi(x,y,-w_{\min})$: by Lemma \ref{lem_height}, both $\rho(\quer_1)$ and $\rho(\quer_2)$ are on the boundary. Moreover, any other homothet $\quert\in S$ will satisfy $\quert \not\in \pi(x,y, -w_{\min})$ (since other ranges have larger weighted distance, which is equivalent to having smaller $z_{\quert}(p)$). That is, the cone $\pi(x,y,-w(\quer_1))$ contains exactly points $\rho(\quer_1)$ and $\rho(\quer_2)$. Moreover, no other point of $\rho(S)$ will be in the cone,  hence $\quer_1\quer_2\in E$. 

The other inclusion is shown analogously: let $\quer_1\quer_2$ be two ranges such that $\quer_1\quer_2\in E$. Let $(x,y,z)\in \Re^3$ be the apex of the minimal cone (with respect to inclusions) such that $\pi(x,y,z)\cap \rho(S)=\{\rho(\quer_1),\rho(\quer_2)\}$. Since $\pi(x,y,z)$ is minimal, both $\rho(\quer_1)$ and $\rho(\quer_2)$ must be on the boundary of the cone. In particular, $z_{\quer_1}(x,y)=z_{\quer_2}(x,y)$ and other ranges satisfy $z_{\quert}(x,y)<z_{\quer_1}(x,y)$ (for all other $\quert\in S$). Using again Lemma \ref{lem_height}, this is equivalent to the fact that $p=(x,y)$ is equidistant to sites $c_1$, $c_2$, and all other sites have strictly larger distance.
\qed\end{proof}

\paragraph{Coloring.}
As an application of the above construction, we show how to solve the dual coloring problem. By Lemma \ref{lem_dualplanar}, we already know that $G_\pi(\rho(S))$ is 4-colorable. 
 For any point $p\in\Re^2$, let $S_p$ be the set of ranges containing $p$ (i.e., $S_p=\{\quert\in S : p\in \quert\}$). 

\begin{lemma}\label{lem_enough}
For any $p\in \Re^2$ such that $|S_p| \geq 2$, there exist $\quer_1,\quer_2 \in S_p$ such that $\quer_1\quer_2\in E(S)$.  
\end{lemma}
\begin{proof}
From the second property of Lemma~\ref{lem_dual}, the number of points of $\rho (S)$ contained in the cone $\pi(p,0)$ is the number of ranges of $S$ containing $p$. The proof is now analogous to Lemma \ref{lem_edge}, where the shrinking operation is replaced by a vertical lifting of the cone. Let $z_0\geq 0$ be the largest value such that the cone $\pi(p,z_0)$ has two or more points of $\rho(S)$. If $\pi(p,z_0)$ contains exactly two points we are done, hence it only remains to treat the degeneracies. Remember that in such a case, there must be at least two points on the boundary of $\pi(p,z_0)$ (and possibly a point in its interior). 

If there is a point $\rho(\quer_1)$ in the interior, we select a second point $\rho(\quer_2)$ on the boundary of $\pi(p,z_0)$ and translate the apex of the cone towards $\rho(\quer_2)$. Note that when the apex is located at $\rho(\quer_2)$, the point $\rho(\quer_1)$ cannot be in the cone (since it would imply that $\quer_1 \subseteq \quer_2$, and we assumed otherwise). Thus, at some point in the translation $\rho(\quer_1)$ reaches the boundary of the cone. During this translation process, points that were on the boundary of $\pi(p,z_0)$ have either remained on the boundary or have left the cone. In either case, we obtain a cone $\pi(p',z')\subseteq \pi(z,0)$ with two (or more) vertices of $\rho(S)$ on its boundary and none in its interior.

If it contains exactly two points of $\rho(S)$ we are done. Otherwise, by duality of Lemma \ref{lem_dualVorDel}, $p'$ is a vertex of the weighted Voronoi diagram. We pick a point $p''$ in an edge $e$ of the Voronoi diagram incident to $p'$. Since $p'$ is in an edge of the Voronoi diagram, it is equidistant to two ranges $\quer_1,\quer_2 \in S_p$. Hence, by Lemma \ref{lem_dualVorDel}, when we do the lifting operation on $p'$ we will obtain a cone that exactly contains $\quer_1,\quer_2 \in S_p$ on its boundary and no other point in its interior.
\qed\end{proof}


\begin{theorem}\label{theo_dual}
$\bar{c}_\quer(2) \leq 4$.
\end{theorem}
\begin{proof}
Let $I\subset S$ be the set of homothets included in other homothets of $S$. Recall that we initially assumed that $I$ was empty. Thus, to finish the proof it only remains to study the $I\neq \emptyset$ case. First we color $S\setminus I$ with 4 colors, using a 4-coloring of $G_\pi(\rho(S))$. This is possible from Lemma~\ref{lem_dualplanar}. Then, for each homothet $\quert$ of $I$ there exists one homothet $\quertt$ in $S \setminus I$ that contains it. We assign to $\quert$ any color different than the one assigned to $\quertt$. Any point $p\in \quert$ will also satisfy $p\in\quertt$, since $\quert\subseteq\quertt$, hence $p$ will be covered by two ranges of different colors.
\qed\end{proof}

\section{Coloring Three-dimensional Hypergraphs}\label{sec_3D}


Lemma~\ref{lem_dualplanar} actually generalizes the ``easy" direction of Schnyder's characterization of planar graphs. We first give a brief overview of this fundamental result. The {\em vertex-edge incidence poset} of a hypergraph $G=(V,E)$ is a bipartite poset $P=(V\cup E,\preceq_P)$, such that $e\preceq_P v$ if and only if $e\in E$, $v\in V$, and $v\in e$. The {\em dimension} of a poset $P=(S,\preceq_P)$ is the smallest $d$ such that there exists an injective mapping $f:S\to\Re^d$, such that $u\preceq_P v$ if and only if $f(u)\leq f(v)$, where the order $\leq$ is the componentwise partial order on $d$-dimensional vectors. When $P$ is the vertex-edge incidence poset of a hypergraph $G$, we will refer to this mapping as a {\em realizer} of $G$, and to $d$ as its {\em dimension}.

There exists a relation between the dimension of a graph and its chromatic number. For example, the graphs of dimension $2$ or less are subgraphs of the path, hence are 2-colorable. Schnyder pointed out that all 4-colorable graphs have dimension at most $4$ \cite{schnyder}, and completely characterized the graphs whose incidence poset has dimension $3$:
\begin{theorem}[\cite{schnyder}]
A graph is planar if and only if its dimension is at most three.
\end{theorem}
The ``easy" direction of Schnyder's theorem states that every graph of dimension at most three is planar. The non-crossing drawing that is considered in one of the proofs is similar to ours, and simply consists of, for every edge $e=uv$, projecting the two line segments $f(e)f(u)$, and $f(e)f(v)$ onto the plane $x+y+z=0$~\cite{trotter,BD81}. It is easy to see that octants in $\Re^3$ satisfy the equivalent of our Lemma~\ref{lem_enough} (by translating the apex of the octant with vector $(1,1,1)$ for example). Combining this result with the Four Color Theorem gives the following result.
\begin{lemma}\label{lem_3dim}
Every hypergraph of dimension at most three is 4-colorable.
\end{lemma}

\subsection{Upper bounds for three-dimensional hypergraphs}\label{sec_3dim}
We now adapt the above result for higher values of $k$. That is, we are given a three-dimensional hypergraph $G=(V,H)$ and a constant $k\geq 2$. We would like to color the vertices of $G$ such that any hyperedge $e\in H$ contains at least $\min\{|e|,k\}$ vertices with different colors. 
 We denote by $c_3(k)$ the minimum number of colors so that any three-dimensional hypergraph can be suitably colored. Note that the problem is self-dual: any instance of the dual coloring problem can be transformed into a primal coloring problem by symmetry.

For simplicity, we assume that no two vertices of $V$ in the realizer share an $x$, $y$ or $z$ coordinate. This can be obtained by making a symbolic perturbation of the point set in $\Re^3$. Recall that, from the definition of the realizer, the point $q_e$ dominates $u\in S$ if and only if $u\in e$. For any hyperedge $e\in H$, there exist many points in $\Re^3$ that dominate the points of $e$. We also assume that hyperedge $e$ is mapped to the minimal point $q_e\in\Re^3$, obtained by translating $q_e$ in each of the three coordinates until a point of hits the boundary of the upper octant whose apex is $q_e$. 

For any hyperedge $e\in H$, we define the $x$-extreme of $e$ as the point $x(e)\in e$ whose image has smallest $x$-coordinate. Analogously we define the $y$ and $z$-extremes and denote them $y(e)$ and $z(e)$, respectively. We say that a hyperedge $e$ is {\em extreme} if two extremes of $e$ are equal.

\begin{lemma}\label{lem_boundeg}
For any $k\geq 2$, $G$ has up to $3n$ extreme hyperedges of size exactly $k$.
\end{lemma}
\begin{proof}
We charge any extreme hyperedge to the point that is repeated. By the pigeonhole principle, if a point is charged more than three times, there exist two extreme hyperedges $e_1,e_2$ of size exactly $k$ that charge on the same coordinates. Without loss of generality, we have $x(e_1)=x(e_2)$ and $y(e_1)=y(e_2)$. Let $q_1$ and $q_2$ be the mappings of $e_1$ and $e_2$, respectively. By hypothesis, the $x$ and $y$ coordinates of $q_1$ and $q_2$ are equal. Without loss of generality, we assume that $q_1$ has higher $z$ coordinate than $q_2$. In particular, we have $q_1\subset q_2$. Since both have size $k$, we obtain $e_1=e_2$. 
\qed\end{proof}

Let $S$ be the the 3-dimensional realizer of the vertices of $G$. For simplicity, we assume that $G$ is maximal. That is, for any $e\subseteq S$, we have $e\in H$ if and only if there exists a point $q_e\in \Re^3$ dominating exactly $e$. Since we are only adding hyperedges to $G$, any coloring of this graph is a valid coloring of $G$. 

For any $2\leq k\leq n$, we define the graph $G_k(S)=(S,E_k)$, where for any $u,v\in S$ we have $uv\in E_k$ if and only if there exists a point $q\in \Re^3$ that dominates $u,v$ and at most $k-2$ other points of $S$ (that is, we replace hyperedges of $G$ whose size is at most $k$ by cliques). The main property of this graph is that any proper coloring of $G_k(S)$ induces a polychromatic coloring of $G$. Using Lemma~\ref{lem_boundeg}, we can bound the number of edges of $G_k(S)$.
\begin{lemma}\label{lem_edges}
For any set $S$ of points and $2\leq k\leq n$, graph $G_k(S)$ has at most $3(k-1)n-6$ edges.
\end{lemma}
\begin{proof}
The claim is true for $k=2$ from Schnyder's characterization. Notice that $E_{k-1}\subseteq E_k$, thus it suffices to bound the total number of edges $uv\in E_k\setminus E_{k-1}$. By definition of $G_k$ and $G_{k-1}$, there must exist a hyperedge $e$ of size exactly $k$ such that $u,v\in e$. In the three-dimensional realizer, this corresponds to a point $q_e\in \Re^3$ that dominates $u,v$ and $k-2>0$ other points of $S$.

We translate the point $q_e$ upward on the $x$ coordinate until it dominates only $k-1$ points. By definition, the first point to leave must be the $x$-extreme point $x(e)$. After this translation we obtain point $q'_e$ that dominates $k-1$ points. All these points will form a clique in $E_{k-1}$. Since $uv\not\in E_{k-1}$,  we either have $u=x(e)$ or $v=x(e)$. We repeat the same reasoning translating in the $y$ and $z$ coordinates instead and, combined with the fact that a point cannot be extreme in the three directions, either $uv\in E_{k-1}$ or $u$ and $v$ are the only two extremes of $e$. In  particular, the hyperedge $e$ is extreme. From Lemma \ref{lem_boundeg} we know that this case can occur at most $3n$ times, hence we obtain the recurrence $|E_k|\leq |E_{k-1}|+3n$. 
\qed\end{proof}

\begin{theorem}\label{theo_colz}
For any $k\geq 2$, we have $c_3(k)\leq 6(k-1)$.
\end{theorem} 
\begin{proof}
From Lemma \ref{lem_edges} and the handshake lemma, the average degree of $G$ is strictly smaller than $6(k-1)$. In particular, there must exist a vertex whose degree is at most $6(k-1)-1$. Moreover, this property is also satisfied by any induced subgraph, as any edge $(u,v)\in E_k$ is an edge of $G_k(S\setminus\{w\})$, $\forall w\neq u,v$. Hence, for any $S'\subseteq S$, the induced subgraph $G_k(S) \setminus S'$ is a subgraph of $G_k(S\setminus\{S'\})$. In particular, the graph $G_k(S)$ is $(6(k-1)-1)$-degenerate, and can therefore be colored with $6(k-1)$ colors.
\qed\end{proof}

Note that dual hypergraphs induced by collections of homothetic triangles have dimension at most 3, so our result directly applies.
\begin{corollary}\label{cor_triangle}
For any $k\geq 3$, any set $S$ of homothets of a triangle can be colored with $6(k-1)$ colors so that any point $p\in\Re^2$ covered by $r$ homothets is covered by $\min\{r,k\}$ homothets with distinct colors. 
\end{corollary}


\section{Lower Bounds}\label{sec_lb}
We now give a lower bound on $c_\quer(k)$. The normal vector of $\quer$ at the boundary point $p$ is the unique unit vector orthogonal to the halfplane tangent to $\quer$ at $p$, if it is well-defined. We say that a range has $m$ distinct directions if there exist $m$ different points with defined, pairwise linearly independent normal vectors. 
\begin{lemma}
\label{lem_lowerprim}
Any range $\quer$ with at least three distinct directions satisfies $c_\quer(k)\geq 4\lfloor k/2 \rfloor$ and $\bar{c}_\quer(k)\geq 4\lfloor k/2 \rfloor$.
\end{lemma}
\begin{proof}
We first show that $c_\quer(2)\geq 4$. Scale $\quer$ by a large enough value so that it essentially becomes a halfplane. By hypothesis, we can obtain halfplane ranges with three different orientations. By making an affine transformation to the problem instance, we can assume that the halfplanes are of the form $x\geq c$, $y\geq \sqrt{3}x+c$ or $y\leq \sqrt{3}x+c$ for any constant $c\in \Re$ (i.e. the directions of the equilateral triangle). Let $\Delta$ be the largest equilateral triangle with a side parallel to the abscissa that can be circumscribed in $\quer$. Let $p_1,p_2, p_3$ and $p_4$ be the vertices and the incenter of $\Delta$, respectively (see Figure \ref{fig_ranges}). Note that any two points of $\{p_1,p_2,p_3,p_4\}$ can be selected with the appropriate halfplane range, hence any valid coloring must assign different colors to the four points. The proof of the dual bound is analogous: it suffices to consider the ranges that contain exactly two points of $\{p_1,p_2,p_3,p_4\}$.

For higher values of $k$ it suffices to replace each point $p_i$ for a cluster of $\lfloor k/2 \rfloor$ points. That is, we have $4\lfloor k/2 \rfloor$ points clustered into four groups so that any two groups can be covered by one range. By the pigeonhole principle, any coloring that uses strictly less than $4\lfloor k/2 \rfloor$ colors must have two points with the same color. The range containing them (and any other $k-2$ points) will have at most $k-1$ colors, hence will not be polychromatic.

\qed\end{proof}
Observe that parallelograms are the only ranges that do not have three or more distinct normal directions (in this case, we can show a weaker $3\lfloor k/2 \rfloor$ lower bound). In particular, the results of Sections \ref{sec_primal} and \ref{sec_dual} are tight for any range other than a parallelogram.  
Also notice that, since triangle containment posets are 3-dimensional, the lower bound also applies to $c_3(k)$.




\section{Applications to other coloring problems}\label{sec_appl}

\paragraph{Conflict-free colorings.}
A coloring of a hypergraph is said to be {\em conflict-free} if, for every hyperedge $e$ there is a vertex $v\in e$ whose color is distinct from all other vertices of $e$. Even {\em et al.}~\cite{shakharcf} gave an algorithm for finding such a coloring. Their method repeatedly colors (in the polychromatic sense) the input hypergraph with few colors, and removes the largest color class. By repeating this process iteratively a conflict-free coloring is obtained. The number of colors is at most $\log_{\frac c{c-1}} n$, where $n$ is the number of vertices, and $c$ is the maximum number of colors used at each iteration. Our 4-colorability proof of Theorem~\ref{theo_dual} is constructive and can be computed in $O(n^2)$ time. Hence, by combining both results we obtain the following corollary.
\begin{corollary}
Any dual hypergraph induced by a finite set of $n$ homothets of a compact and convex body in the plane has a conflict-free coloring using at most $\log_{4/3} n \leq 2.41\log_2 n$ colors. Furthermore, such a coloring can be found in $O(n^2\log n )$ time.
\end{corollary}

\paragraph{$k$-strong conflict-free colorings.}
Abellanas {\em et al.}~\cite{ABGHNR09} introduced the notion of $k$-strong conflict free colorings, in which every hyperedge $e$ has $\min\{|e|,k\}$ vertices with a unique color. Conflict-free colorings are $k$-strong conflict-free colorings for $k=1$. Recently, Horev, Krakovski, and Smorodinsky~\cite{HKS10} showed how to find $k$-strong conflict-free colorings by iteratively removing the largest color class of a polychromatic coloring with $c(k)$ colors. Again, combining this result with Theorem~\ref{theo_colz} yields the following corollary. 
\begin{corollary}
Any dual hypergraph induced by a finite set of $n$ homothets of a compact and convex body in the plane has a $k$-strong conflict-free coloring using at most $\log_{(1+\frac{1}{6(k-1)})}n$ colors.
\end{corollary}
\paragraph{Choosability.}
Cheilaris and Smorodinsky~\cite{CS10} introduced the notion of choosability in geometric hypergraphs. A hypergraph with vertex set $V$ is said to be $k$-choosable whenever for any collection $\{L_v\}_{v\in V}$ of subsets of positive integers of size at least $k$, the hypergraph admits a proper coloring, where the color of vertex $v$ is chosen from $L_v$. Our construction of Section~\ref{sec_dual} provides a planar graph, and planar graphs are known to be 5-choosable. This directly yields the following result.
\begin{corollary}
Any dual hypergraph induced by a finite set of homothets of a convex body in the plane is 5-choosable.
\end{corollary}

\bibliographystyle{plain}
\bibliography{fourCol}
\end{document}